\documentclass{amsart}
\usepackage{amssymb}
\usepackage{graphicx}
\usepackage{amscd}

\newtheorem{theorem}{Theorem}
\theoremstyle{plain}

\newtheorem{lemma}{Lemma}

\newtheorem{proposition}{Proposition}
\newtheorem{remark}{Remark}

\newtheorem{application}{Application}

\numberwithin{equation}{section}

\begin{document}

\begin{center}
\textbf{A SIMPLE NOTE ON SOME EMPIRICAL STOCHASTIC PROCESS AS A TOOL IN
UNIFORM L-STATISTICS WEAK LAWS} \\[0pt]
Gane Samb LO\\[0pt]
\end{center}

\bigskip

\begin{center}
\textbf{LERSTAD}, Universit\'e Gaston Berger, Saint-Louis, SENEGAL.\\[0pt]
\textbf{LSTA}, Universit\'e Pierre et Marie Curie, Paris, FRANCE.\\[0pt]
\bigskip ganesamblo@ufrsat.org, ganesamblo@yahoo.com
\end{center}

\bigskip \bigskip 

\begin{center}
\textbf{Abstract : }
\end{center}

\noindent In this paper, we are concerned with the stochastic process 
\begin{equation}
\beta _{n}(q_{t},t)=\beta _{n}(t)=\frac{1}{\sqrt{n}}\sum_{j=1}^{n}\left\{
G_{t,n}(Y(t))-G_{t}(Y_{j}(t))\right\} q_{t}(Y_{j}(t)),  \tag{A}
\end{equation}
where for $n\geq1$ and $T>0$, the sequences $\{Y_{1}(t),Y_{2}(t),...,Y_{n}(t),t\in [0,T]\}$
are independant observations of some real stochastic process ${Y(t),t\in
[0,T]}$, for each $t \in [0,T]$, $G_{t}$ is the distribution function of $%
Y(t)$ and $G_{t,n}$ is the empirical distribution function based on $%
Y_{1}(t),Y_{2}(t),...,Y_{n}(t)$, and finally $q_{t}$ is a bounded real
fonction defined on $\mathbb{R}$. This process appears when investigating
some time-dependent L-Statistics which are expressed as a function of some
functional empirical process and the process (A). Since the functional
empirical process is widely investigated in the literature, the process
reveals itself as an important key for L-Statistics laws. In this paper, we
state an extended study of this process, give complete calculations of the
first moments, the covariance function and find conditions for asymptotic
tightness.\newline

\section{\protect\bigskip Introduction}

\label{sec1}

In this paper, we are concerned with the uniform weak laws of a special
process occuring in some research areas like Actuarial Sciences when measuring
heavy losses, Welfare Sciences when measuring inequality coefficients and
poverty indices. As well, it may be applied for general L-statistics. In
order to define it, let $n\geq 1$ be a positive integer and $%
Y_{1},Y_{2},...,Y_{n}$ independent and identically distributed random
variables with values in $\ell ^{\infty }([0,T]),$ the space of real bounded
functions defined on time space $[0,T],$ where $T$ is a fixed positive real
number. This means that the observations depend on the time $t\in \lbrack
0,T],$ so that we may also write them in the form 
\begin{equation*}
\{Y_{1}(t),Y_{2}(t),...,Y_{n}(t),t\in \lbrack 0,T]\}
\end{equation*}
and we represent the order statistics, when needed, by $Y_{1,n}(t)\leq
Y_{2,n}(t)\leq ...\leq Y_{n,n}(t).$ Now let $k\geq 1$ and $%
0<t_{1}<t_{2}<...<t_{k}\leq T,$ $G_{t_{1},t_{2},...,t_{k}}$ will stand for
the distribution function of $%
(Y_{j}(t_{1}),Y_{j}(t_{2}),...,Y_{j}(t_{k}))^{t}.$ Also, for each $t\in
\lbrack 0,T],$ we denote by $G_{t,n}$ the empirical distribution function
based on the sample $Y_{1}(t),Y_{2}(t),...,Y_{n}(t),$ that is, for each $x\in 
\mathbb{R},$%
\begin{equation*}
nG_{t,n}(x)=\sum_{j=1}^{n}1_{(Y_{j}(t)\leq x)}.
\end{equation*}
From now, we suppose that all the random variables used here are defined on
the same probability space $(\Omega ,\mathcal{A},\mathbb{P}).$ We are now
able to introduce the process 
\begin{equation}
\beta _{n}(q_{t},t)=\beta _{n}(t)=\frac{1}{\sqrt{n}}\sum_{j=1}^{n}\left\{
G_{t,n}(Y(t)-G_{t}(Y_{j}(t)\right\} q_{t}(Y_{j}(t)),  \label{def01}
\end{equation}
where for each $t\in \lbrack 0,1],$ $q_{t}:\mathbb{R}\longmapsto \mathbb{R}$
is a measurable bounded function. For $q\equiv 1,$ we write it $B_{n}^{\ast
}(t)=\beta _{n}(1,t)$ and called it as the simple process. This process $%
\{\beta _{n}(t),t\in \lbrack 0,T]\}$ may appear when dealing with
time-dependant L-Statistics of the form 
\begin{equation}
J_{n}(t)=\frac{1}{n}\sum_{j=1}^{Q_{n}(t)}c(j/n)q_{0}(Y_{j,n}(t)),
\label{def02}
\end{equation}
where $c(\cdot )$ (resp. $q_{0}(\cdot ))$ is a function defined on $[0,1]$
(resp. $\mathbb{R}$) and where for each fixed $t\in [0,T]$, $Z(t)>0$ is some
threshold such that $Y_{Q_{n},n}(t)\leq Z(t) <Y_{Q_{n}+1,n}(t)$. By denoting $%
R_{j,n}(t)$ the rank statistics of $Y_{j}(t),$ (\ref{def02}) may be written,
when the distribution functions $G_{t}$ are continuous, as 
\begin{equation*}
J_{n}(t)=\frac{1}{n}\sum_{j=1}^{n}c(R_{j,n}(t)/n)q(Y_{j}(t))\mathbb{I}%
(Y_{j}(t)\leq Z(t))
\end{equation*}
\begin{equation*}
=\frac{1}{n}\sum_{j=1}^{n}c(G_{t,n}(Y_{j}(t))q_{1}(Y_{j}(t)),
\end{equation*}
where $q_{1}(Y(t))=q_{0}(Y(t))\mathbb{I}(Y(t)\leq Z(t))$. Under some
conditions (see \cite{lo1}), (\ref{def02}) may be uniformly approximated by
the representation, as $n\rightarrow \infty ,$ 
\begin{equation*}
J_{n}(t)=\frac{1}{n}\sum_{j=1}^{n}c(G_{t}(Y_{j}(t))q_{1}(Y_{j}(t))
\end{equation*}
\begin{equation*}
+\frac{1}{n}\sum_{j=1}^{n}\left\{ G_{t,n}(Y(t)-G_{t}(Y_{j}(t)\right\}
c^{\prime }(G_{t}(Y_{j}(t))q_{1}(Y_{j}(t))+o^{\ast}_{P}(n^{-1/2}),
\end{equation*}

\noindent where $c^{\prime }$ is the derivative function of $c$, and $%
u^{\ast}_{n}=o^{\ast}_{P}(1)$ stands for the convergence to zero in
outer-probability, that is there exists a sequence of random variables $%
u_{n} $ converging to zero in probability as $n\rightarrow +\infty$ and $%
\|u^{\ast}_{n}\| \leq \|u_{n}\|$ for $n\geq 1$. Putting 
\begin{equation*}
J(t)=\mathbb{E}c(G_{t}(Y_{j}(t))q_{1}(Y_{j}(t))=\int_{\mathbb{R}%
}c(G_{t}(y))q_{1}(y)dG_{t}(y),
\end{equation*}
we have, for $q_{t}(\cdot )=c^{\prime }(G_{t}(\cdot )q_{1}(\cdot ),$ as $%
n\rightarrow \infty ,$%
\begin{equation}
\sqrt{n}(J_{n}(t)-J(t))=\alpha _{n}(t)+\beta _{n}(q_{t},t)+o_{P}(1),
\label{def03}
\end{equation}
where 
\begin{equation*}
\alpha _{n}(t)=\frac{1}{\sqrt{n}}\sum_{j=1}^{n}\left\{
c(G_{t}(Y_{j}(t))q_{1}(Y_{j}(t))-\mathbb{E}c(G_{t}(Y_{j}(t))q_{1}(Y_{j}(t))%
\right\} ,
\end{equation*}

\noindent and this is nothing else but the functional empirical process $\mathbb{G}%
_{n} $ so that 
\begin{equation*}
\alpha _{n}(t)=\mathbb{G}_{n}(g_{t})=\frac{1}{\sqrt{n}}\sum_{j=1}^{n}\left\{
g_{t}(Y_{j})-\mathbb{E}g_{t}(Y_{j})\right\} ,
\end{equation*}
where $g_{t}$ is the real function defined on $\ell ^{\infty }([0,T])$
satisfying 
\begin{equation*}
g_{t}(x)=c(G_{t}(x(t))q_{1}(x(t)), x\in \ell ^{\infty }([0,T]).
\end{equation*}
Statistics like (\ref{def02}) thus are present in many situations in
connection with L-Statistics (see \cite{hemlers1}, \cite{hemlers2}, \cite%
{hemlers3}) and naturally occur is Acturial Sciences and in inequality
measures (see \cite{zikitis}), and more recently in poverty measures (see 
\cite{lo1}, \cite{lo2}). In all these fields, we may be faced not to find
simple asymptotic normality results, but to derive uniform asymptotic laws
for the time-dependant statistics (with the parameter $t\in \lbrack 0,T] $)
and functional asymptotic laws with respect to the class of functions $%
\mathcal{F}=\{(g_{t},q_{t}),t\in \lbrack 0,T]\}.$

\bigskip

This motivated us to undertake a special study of $\beta _{n}$ and its
connection with the empirical process as general key tools. This study needs
much calculations that may be superfluous in each particular application.
We thus aim to characterize this process here and present our results as
general tools to be used further in statistical works as packages. In all
the paper, we suppose that the distribution functions $G_{t}$ are continuous
and increasing.

\bigskip

Since the calculations related to this study are tremendous, we are going to
give here the characteristics of the process. Examples of computations that
lead to the results stated here are given in the beginning of the proof of the first theorem
while the full paper are given in \cite{lo3}.\\

\bigskip 

The paper is organized as follows. We entirely describe the weak law the process
in Section 2. In Section 3, the weak law of the sum of a process of type (\ref{def01}) with a functional empirical process is given while Section 3 is devoted to the weak law of a couple of statistics of type (\ref{def01}). The paper is finished by a conclusion.

\section{Law of the general process} \label{sec2}

We now consider the process 
\begin{equation*}
\beta _{n}^{\ast }(t)=\sqrt{n}\beta _{n}(t)=\sum_{j=1}^{n}\left\{
G_{t,n}(Y_{j}(t))-G_{t}(Y_{j}(t))\right\} q_{t}(Y_{j}(t)).
\end{equation*}

\noindent Before we present our main result, define

\begin{equation*}
g(q,t,s)=\int \left( \int_{x\geq u}q_{t}(x)G_{t}(x)\right) \left(
\int_{y\geq v}q_{s}(y)G_{t}(y)\right) dG_{t,s}(u,v),
\end{equation*}

\begin{equation*}
c_{2}(t)=\int \left( \int_{x\geq u}q_{t}(x)G_{t}(x)\right) ^{2}dG_{t}(u)
\end{equation*}
and this convention, for a function $h,$%
\begin{equation*}
\mathbb{E}_{t}h=\int h(u)dG_{t}(u)
\end{equation*}

\begin{theorem}
If there is a universal constant $K_{0,}$ such that there exists $\delta
>0, $%
\begin{equation*}
\left\vert s-t\right\vert \leq \delta \Longrightarrow \left\vert 2(c_{2}(t)-g(q,t,s)) \right.
\end{equation*}

\begin{equation}
\left. +\left\{ (\mathbb{E}_{t}G_{t}q_{t})(\mathbb{E}%
_{s}G_{s}q_{t})-(\mathbb{E}_{t}G_{t}q_{t})^{2}\right\} \right\vert \leq 
\frac{3}{2}K_{0}\left\vert s-t\right\vert ^{1+r},  \label{unif04}
\end{equation}

\noindent then $\{\beta _{n}(t),0\leq t\leq T\}$ converges to a $\ell ^{\infty
}([0,T])-$Gaussian process with covariance function 
\begin{equation*}
\Gamma _{1}(q_{t},q_{s},s,t)=g(q,t,s)-(\mathbb{E}_{t}G_{t}q_{t})(\mathbb{E}%
_{s}G_{s}q_{s}).
\end{equation*}
\end{theorem}

\bigskip

\begin{remark}
As announced, we will give in the beginning of the proof of this theorem
examples of computations needed in proving the results of these paper. Full,
detailed and complete ones are stated in \cite{lo3}.
\end{remark}

\begin{proof}
Let 
\begin{equation*}
\beta _{n}^{\ast }(t)=\sqrt{n}\beta _{n}(t).
\end{equation*}%
We  begin to calculate the two first moments and the covariance function.

\bigskip \noindent \textbf{Mean calculation}. One has 
\begin{equation*}
\mathbb{E}\beta _{n}^{\ast }(t)=\mathbb{E}%
\sum_{j=1}^{n}G_{t,n}(Y_{j}(t))q_{t}(Y_{j}(t))-n(\mathbb{E}%
q_{t}(Y(t))(G_{t}(Y(t)).
\end{equation*}%
But 
\begin{equation*}
nG_{t,n}(Y_{j}(t))q_{t}(Y_{j}(t))=q_{t}(Y_{j}(t))+\sum_{h\neq j}1_{\left(
Y_{h}(t)\leq Y_{j}(t)\right) }q_{t}(Y_{j}(t))
\end{equation*}%
and 
\begin{equation*}
\mathbb{E}nG_{t,n}(Y_{j}(t))q_{t}(Y_{j}(t))=\mathbb{E}q_{t}(Y(t))+(n-1)\int
q_{t}(u)dG_{t}(u)\int_{x\geq u}dG_{t}(x)
\end{equation*}%
\begin{equation*}
=\mathbb{E}_{t}q_{t}+(n-1)\int G_{t}(u)q_{t}(u)dG_{t}(u).
\end{equation*}%
Recall the convention $\mathbb{E}_{t}b=\mathbb{E}(b(Y(t))).$ We get 
\begin{equation*}
\mathbb{E}G_{t,n}(Y_{j}(t))q_{t}(Y_{j}(t))=\frac{\mathbb{E}_{t}q_{t}-\mathbb{%
E}_{_{t}}q_{t}G_{t}}{n}+\mathbb{E}_{t}q_{t}G_{t}.
\end{equation*}%
This gives 
\begin{equation*}
\mathbb{E}\beta _{n}^{\ast }(t)=\mathbb{E}_{_{t}}q_{t}-\mathbb{E}%
_{_{t}}q_{t}G_{t}
\end{equation*}%
and 
\begin{equation*}
\mathbb{E}\beta _{n}(t)=\left( \mathbb{E}_{_{t}}q_{t}-\mathbb{E}%
_{_{t}}q_{t}G_{t}\right) /\sqrt{n}\rightarrow 0.
\end{equation*}%
\noindent \textbf{Variance calculation}. Direct calculations like the
previous give : 
\begin{equation*}
\mathbb{E}\beta _{n}(t)^{2}=c_{2}(t)-(\mathbb{E}_{t}G_{t}q_{t})^{2}+\frac{%
K_{1}(t,s)}{n},
\end{equation*}%
where $K_{1}(t,s)$ is uniformly bounded. Before we arrive at the covariance
function. We should observe that for $q_{t}=1,$ then $c_{2}=1/3,(\mathbb{E}%
_{t}G_{t}q)^{2}=1/4$ and 
\begin{equation*}
c_{2}(t)-(\mathbb{E}_{t}G_{t}q_{t})^{2}=1/12.
\end{equation*}

\bigskip

\noindent \textbf{Covariance calculations. }We also have%
\begin{equation*}
\mathbb{E}\beta _{n}(t)\beta _{n}(s)=g(q,t,s)-(\mathbb{E}_{t}G_{t}q_{t})(%
\mathbb{E}_{s}G_{s}q_{s})+\frac{K_{2}(n,t,s)}{n}.
\end{equation*}%
\begin{equation*}
=\Gamma _{1}(q_{t},q_{s},t,s)+\frac{K_{2}(n,t,s)}{n},
\end{equation*}%
where $K_{2}(n,t,s)$ is uniformly bounded in $(n,t,s).$ We finish to remark
that for $s=t,$ we get 
\begin{equation*}
\mathbb{E}\beta _{n}(t)^{2}\sim c_{2}(t)-(E_{t}G_{t}q)^{2}.
\end{equation*}%
We now consider the increments of $\beta _{n}(t).$

\bigskip \noindent \textbf{Increments calculations.}

\noindent Recall that 
\begin{equation*}
E\beta _{n}(t)^{2}=c_{2}(t)-(E_{t}G_{t}q)^{2}+\frac{K_{1}(n,t)}{n}.
\end{equation*}
This gives 
\begin{equation*}
E\left( \beta _{n}(t)-\beta _{n}(s)\right) ^{2}=2(c_{2}(t)-g(q,t,s))
\end{equation*}

\begin{equation}
+\left\{(E_{t}G_{t}q)(E_{s}G_{s}q)-(E_{t}G_{t}q)^{2}\right\} +\frac{K_{3}(n,t,s)}{n}.
\label{unif03}
\end{equation}

\bigskip \noindent \textbf{Proofs of the weak convergence.}

We always begin to show the weak convergence of the finite-distribution of $%
\beta _{n}(\cdot )$ that is 
\begin{equation*}
\beta _{n}(t_{1},...,t_{k},a)=\sum_{j=1}^{k}\alpha _{j}\beta _{n}(t_{j})=%
\frac{1}{\sqrt{n}}\sum_{s=1}^{k}a_{s}\sum_{j=1}^{n}\left\{
G_{t_{s},n}(Y_{j}(t_{s}))-G(Y_{t_{s}})\right\} q_{t_{s}}(Y_{j}(t_{s})).
\end{equation*}%
$0<t_{0}<t_{1}<...<t_{k}\leq T,$ $a=(a_{1},...,a_{k})^{t}\in \mathbb{R}^{k}.$
\ We have 
\begin{equation}
\beta _{n}(t)=\int_{0}^{1}\sqrt{n}%
(s-V_{t,n}(s))q_{t}(G_{t}^{-1}(s))ds+O_{P}(1/\sqrt{n})=N_{n}^{\ast
}(q_{t},t)+O_{P}(1/\sqrt{n}).  \label{marge02}
\end{equation}%
The finite distribution is established by using Lemma \ref{lemmatool} below and
its application in section \ref{sectool}. The covariance function of the
limiting process is 
\begin{equation*}
\Gamma _{1}(q_{t_{i}},q_{t_{j}},t_{i},t_{j})=\lim_{n\rightarrow \infty } \mathbb{C}ov(N_{n}^{\ast
}(q_{t_{i}},t_{i}),N_{n}^{\ast }(q_{t_{j}},t_{j}))
\end{equation*}%
which, by (\ref{marge02}), is
\begin{equation*}
\Gamma _{1}(q_{t_{i}},q_{t_{j}},t_{i},t_{j})=\lim_{n\rightarrow \infty }%
\mathbb{C}ov(\beta _{n}(t),\beta _{n}(s))
\end{equation*}%
Finally (\ref{unif04}), (\ref{unif03}) together prove the asymptotic
tightness of $\beta _{n}$ via Lemma 1 in \cite{sall-lo} and Example 2.2.12
in \cite{vaart}.
\end{proof}

\newpage

\section{Addition of the processes and an empirical process}

\label{sec3}

In many situations, the asymptotic law of the studied statistics is achieved
in a sum of our process and an empirical process of the form 
\begin{equation*}
\gamma _{n}=\alpha _{n}+\beta _{n}
\end{equation*}
where 
\begin{equation*}
\gamma _{n}(t)=\frac{1}{\sqrt{n}}\sum_{j}(g_{t}(Y(t))-\eta (t))+\frac{1}{%
\sqrt{n}}\sum_{j}\left\{ G_{t,n}(Y_{j}(t))-G_{t}(Y_{j}(t))\right\}
q_{t}(Y_{j}(t)).
\end{equation*}
In such cases, what is the covariance structure of the limiting process?
We have this

\begin{proposition}
\bigskip If each of the processes $\gamma _{n},\alpha _{n}$ and $\beta _{n}$
converges in finite-distributions and is asymptotically tight, then
covariance function of the limiting Gaussian process of $\gamma _{n}$ is 
\begin{equation*}
\Gamma (t,s)=\Gamma _{1}(q_{t},q_{s},t,s)+\Gamma _{2}(t,s)+\gamma (t,s),
\end{equation*}
with 
\begin{equation*}
\Gamma _{2}(t,s)=\int (\overline{g}_{t}(x)-\eta (t))(\overline{g}%
_{s}(y)-\eta (s))dG_{t,s}(x,y),
\end{equation*}
\begin{equation*}
\Gamma _{2}(q_{t},q_{s},t,s)=g(q,t,s)-(\mathbb{E}_{t}G_{t}q_{t})(\mathbb{E}%
_{s}G_{s}q_{s}),
\end{equation*}
\begin{equation*}
g(q_{t},q_{s},t,s)=\int \left( \int_{x\geq u}q_{t}(u)dG_{t}(u)\right) \left(
\int_{x\geq v}q_{t}(v)dG_{t}(v)\right) dG_{t,s}(u,v)
\end{equation*}
and 
\begin{equation*}
\gamma (t,s)=\gamma _{1}(t,s)+\gamma _{1}(s,t),
\end{equation*}
with 
\begin{equation*}
\gamma _{1}(t,s)=\int \overline{g}_{t}(u)\left( \int_{x\geq
u}q(x)dG_{s}(u)\right) dG_{t,s}(u,v).
\end{equation*}
\end{proposition}

\begin{remark}
We are not interesting here by complete results. We only intend to show
how the process intervenes in general L-Statistics and to give the
covariance function. In each particlucar,we will have to prove the finite-distribution
convergence and the tightness of the components of such processes.
\end{remark}

\begin{proof}
If the hypotheses of the proposition hold, the limiting covariance function
is performed through the formula%
\begin{equation*}
\gamma _{n}(t)\gamma _{n}(s)=\left( \alpha _{n}(t)+\beta _{n}(t)\right)
\left( \alpha _{n}(s)+\beta _{n}(s)\right) 
\end{equation*}%
\begin{equation*}
=\alpha _{n}(t)\alpha _{n}(s)+\alpha _{n}(t)\beta _{n}(s)+\beta
_{n}(t)\left( \alpha _{n}(s)+\beta _{n}(t)\beta _{n}(s)\right) ).
\end{equation*}%
By computing the expectation of each of them, we arrive at 
\begin{equation*}
\gamma _{1}(t,s)=\int \overline{g}_{t}(u)\left( \int_{x\geq
u}q(x)dG_{s}(u)\right) dG_{t,s}(u,v),\gamma (t,s)=\gamma _{1}(t,s)+\gamma
_{1}(s,t).
\end{equation*}
\end{proof}

\section{Covariance function of two processus}

\label{sec4}

In some applications, we may be led to simultaneously consider two or
several processes of the kind (\ref{def01}). In this case, their covariance
function may be useful. Consider 
\begin{equation*}
\beta _{n,2}(t)=\frac{1}{\sqrt{n}}\sum_{j}\left\{
G_{t,n}(Y_{j}(t))-G_{t}(Y_{j}(t))\right\} q_{1,t}(Y_{j}(t))
\end{equation*}
and 
\begin{equation*}
\beta _{n,1}(t)=\frac{1}{\sqrt{n}}\sum_{j}\left\{
G_{t,n}(Y_{j}(t))-G_{t}(Y_{j}(t))\right\} q_{2,t}(Y_{j}(t)).
\end{equation*}
We will have the result

\begin{proposition}
If the two processes are both asymptotically tight and converge in
finite-distribution, then their limiting Gaussian processes have the
following covariance 
\begin{equation*}
\Gamma
_{3}(t,s)=g(q_{1,t},q_{2,s},t,s)-((E_{t}G_{t}q_{1})(E_{s}G_{s}q_{2})+(E_{t}G_{t}q_{1})(E_{s}G_{s}q_{2})),
\end{equation*}
and 
\begin{equation*}
g(q_{1,t},q_{2,s},t,s)=g_{1}(q_{1,t},q_{2,s},t,s)+g_{1}(q_{1,s},q_{2,t},s,t)
\end{equation*}
with 
\begin{equation*}
g_{1}(q_{1,t},q_{2,s},t,s)=\int \left( \int_{x\geq
u}q_{1,t}(u)dG_{t}(u)\right) \left( \int_{x\geq v}q_{2,s}(v)dG_{s}(v)\right)
dG_{t,s}(u,v).
\end{equation*}
\end{proposition}

\section{\protect\bigskip A useful tool}

\label{sectool}

We give here a useful lemma on which, is be based the asymptotic
finite-distribution normality of the processes involved here. It will be
enough to describe it in the two dimensional case. A generelization to the
k-dimensional case is straightforward. We have

\begin{lemma}
\label{lemmatool} Let $(X_{i},Y_{i})$ , $i=1,2,...,$ be \ independent
observations of a random vector (X,Y) with joint distribution function $G(x,y)=P(X\leq x,Y\leq
y)$, and margins $G_{1}(x)=G(x,+\infty )$ and $G_{2}(y)=G(+\infty ,y)$.  Let, for each $n\geq 1,$ $%
\varepsilon _{1,n}$ and $\varepsilon _{2,n}$ be the quantile processes based
respectively on $G_{1}(X_{1}),G_{1}(X_{2}),...,G_{1}(X_{n})$, and on $%
G_{2}(Y_{1}),G_{2}(Y_{2}),...,G_{2}(Y_{n}).$ Then $\varepsilon
_{n}=(\varepsilon _{1,n},\varepsilon _{2,n})$ converges in distribution to a
Gaussian process $\varepsilon =(\varepsilon _{1},\varepsilon _{2})$ in $%
(\ell ^{\infty }([0,1]))^{2}$ such that each $\varepsilon _{i}$ is a
standard Brownian bridge.
\end{lemma}

\begin{proof}
Let for each $n\geq 1$, $\alpha _{1,n}$ and $\alpha _{2,n}$ be the empirical
processes based respectively on $G_{1}(X_{1}),G_{1}(X_{2}),...,G_{1}(X_{n})$
and on $G_{2}(Y_{1}),G_{2}(Y_{2}),...,G_{2}(Y_{n}).$ We have (see \cite%
{shwell}, p.584) that $\alpha _{i,n}(s)=-\varepsilon _{i,n}(s)+o_{P}(1)$
uniformly in $s\in (0,1),$ which gives 
\begin{equation*}
\varepsilon _{n}(s,t)=(\varepsilon _{1,n}(s),\varepsilon _{2,n}(t))=-(\alpha
_{1,n}(s),\alpha _{2,n}(t))=o_{P}(1),
\end{equation*}
uniformly in $(s,t)\in (0,1)^{2}.$ Now let us consider the functional
empirical process $\alpha _{n}$\ based on the $%
Z_{i}=(G_{1}(X_{i}),G_{2}(Y_{i})),$ that is 
\begin{equation*}
\alpha _{n}(f)=\frac{1}{\sqrt{n}}\sum_{j=1}^{n}f(Z_{i})-\mathbb{E}f(Z_{i}),
\end{equation*}
for a real function defined on $(0,1)^{2}$ such that $\mathbb{E}%
f(Z_{i})^{2}<\infty .$ We have by the classical results of empirical process
that $\{\alpha _{n}(f),$ $f\in \mathcal{F}\}$ converges to a Gaussian
process $\{\mathbb{G}(f),f\in \mathcal{F\}}$ whenever $\mathcal{F}$ is a
donsker class. It follows that $\{\alpha _{n}(1_{C}),$ $C\in \mathcal{C}\}$
converges to a Gaussian process $\{\mathbb{G}(1_{C}),C\in \mathcal{C\}}$
whenever $\mathcal{C}$ is a Vapnik-Cervonenkis class ($VP$-class). But $\mathcal{C}%
=\{1_{[0,s]\times \lbrack 0,t\rbrack},(t,s)\in (0,1)^{2}\}$ is a $VP$-class of index
not greater of 2. (see \cite{vaart} for $VP$-classes use to empirical
processes). Thus, putting $f_{s,t}=1_{[0,s]\times \lbrack 0,t]},$, we have
\begin{equation*}
\alpha _{n}(s,t)\equiv \alpha _{n}(f_{s,t})\leadsto \mathbb{G}%
(f_{s,t})\equiv \mathbb{G}(s,t)
\end{equation*}
in $(\ell ^{\infty }([0,1]))^{2}$, where $\leadsto$ stands for the weak convergence.  Now, by using the Skorohod-Wichura-Dudley
Theorem, we are entitled to suppose that we are on a probability space such
that 
\begin{equation*}
\sup_{(s,t)\in (0,1)^{2}}\left| \alpha _{n}(f_{s,t})-\mathbb{G}%
(f_{s,t})\right| \rightarrow _{P}0.
\end{equation*}
Now, put $f_{1,s}=1_{[0,s]\times \lbrack 0,1]},$ $f_{2,t}=1_{[0,1]\times
\lbrack 0,t]},$ $\mathbb{G}_{1}(s)=\mathbb{G}_{1}(f_{1,s})$ and $\mathbb{G}%
_{2}(t)=\mathbb{G}_{1}(f_{2,t}).$ We have 
\begin{equation*}
\alpha _{n}(f_{i,s})=\alpha _{1,n}(s)=\mathbb{G}_{1}(s)+o_{P}(1),
\end{equation*}
uniformly in $s\in (0,1).$ We finally have 
\begin{equation*}
\alpha _{n}(s,t)=(\mathbb{G}_{1}(s),\mathbb{G}_{2}(t))+o_{P}(1),
\end{equation*}
uniformly in $(s,t)\in (0,1)^{2}$. Clearly, $(\mathbb{G}_{1}(s),\mathbb{G}%
_{2}(t))$ is a Gaussian process and each $\mathbb{G}_{i}$ is the standard
Brownian bridge.
\end{proof}

\begin{application}
Let us consider the two-dimensional distribution $\beta _{n}(t_{1},t_{2},a)$%
\ like in (\ref{marge02}), which is 
\begin{equation*} 
\frac{1}{\sqrt{n}}\left\{ {a_{1}\sum_{j=1}^{n}\left\{
G_{t_{1},n}(Y_{j}(t_{1}))-G(Y_{t_{2}})\right\}
q_{t_{1}}(Y_{j}(t_{1}))}\right.
\end{equation*}

\begin{equation*} 
\left. + a_{1}\sum_{j=1}^{n}\left\{
G_{t_{2},n}(Y_{j}(t_{2}))-G(Y_{t_{2}})\right\}
q_{t_{2}}(Y_{j}(t_{2}))\right\} .
\end{equation*}

Using the notations around (\ref{marge02}), we have 
\begin{equation*}
\beta
_{n}(t_{1},t_{2},a)=a_{1}N_{1}(q_{t_{1},}t_{1})+a_{2}N_{1}(q_{t_{2},}t_{2})+o_{P}(1)
\end{equation*}
\begin{equation*}
=\int_{0}^{1}\left\{
a_{1}G_{1}(s)q_{t_{1}}(s)+a_{2}G_{2}(s)q_{t_{2}}(s)\right\} ds+o_{P}(1),
\end{equation*}
\begin{equation*}
\rightarrow N(a_{1},a_{2})=\int_{0}^{1}\left\{
a_{1}G_{1}(s)q_{t_{1}}(s)+a_{2}G_{2}(s)q_{t_{2}}(s)\right\} ds,
\end{equation*}
which is a Gaussian random variable.
\end{application}

\bigskip

\section{Conclusion}
We have entirely described the weak law of empirical stochastic processes like
(\ref{def01}) as well that of such processes and a functional emprical processes. Such results
have potential powerful applications in deriving uniform time-dependent L-statistics as done in
\cite{lo4}, where the time-dependant general poverty index is sutudied. Applications of our results in Actuarial Sciences are under way.

\end{document}